\documentclass[letterpaper, 10 pt, conference]{ieeeconf}  

\IEEEoverridecommandlockouts                              

\usepackage{colortbl} 
\usepackage[table]{xcolor} 
\usepackage{float} 
\usepackage{url}
\usepackage{mathtools} 
\usepackage{graphics} 
\usepackage{epsfig} 
\usepackage{mathptmx} 
\usepackage{times} 
\usepackage{amsmath} 
\usepackage{amssymb}  
\usepackage{cite}
\usepackage{tikz}
\usetikzlibrary{patterns,angles,quotes,arrows}
\newtheorem{remark}{Remark}
\newtheorem{theorem}{Theorem}
\newtheorem{proposition}{Proposition}

\newtheorem{definition}{Definition}
\newtheorem{problem}{Problem}
\newtheorem{cor}{Corollary}

\usepackage{eucal}
\DeclareMathOperator{\interior}{int}
\DeclareMathOperator{\rank}{rank}

\title{\LARGE \bf
Quantitative Resilience of Linear Systems*
}

\author{Jean-Baptiste Bouvier and Melkior Ornik
\thanks{*This work was supported by an Early Stage Innovations grant from NASA’s Space Technology Research Grants Program, grant no. 80NSSC19K0209 and by NASA grant no. 80NSSC21K1030.}
\thanks{Jean-Baptiste Bouvier and Melkior Ornik are with the Department of Aerospace Engineering and the Coordinated Science Laboratory,
        University of Illinois Urbana-Champaign, USA.
        {\tt\small bouvier3@illinois.edu mornik@illinois.edu}}%
}

\begin{document}

\maketitle
\thispagestyle{empty}
\pagestyle{empty}

\begin{abstract}
    Actuator malfunctions may have disastrous consequences for systems not designed to mitigate them. We focus on the loss of control authority over actuators, where some actuators are uncontrolled but remain fully capable. To counteract the undesirable outputs of these malfunctioning actuators, we use real-time measurements and redundant actuators. In this setting, a system that can still reach its target is deemed resilient. 
    To quantify the resilience of a system, we compare the shortest time for the undamaged system to reach the target with the worst-case shortest time for the malfunctioning system to reach the same target, i.e., when the malfunction makes that time the longest.
    Contrary to prior work on driftless linear systems, the absence of analytical expression for time-optimal controls of general linear systems prevents an exact calculation of quantitative resilience. Instead, relying on Lyapunov theory we derive analytical bounds on the nominal and malfunctioning reach times in order to bound quantitative resilience. We illustrate our work on a temperature control system.
\end{abstract}

\section{INTRODUCTION}\label{sec:intro}

Redundancy is the most effective remedy to actuator malfunctions, but it is also the most onerous \cite{NASA_redundancy}, which leads to our questions of interest: can the system still reach its target even after an actuator malfunction? If so, by how much can it delay the system at reaching its target?
We focus on the malfunction consisting in the loss of control authority over actuators \cite{IFAC}. These malfunctioning actuators do not respond to the controller, but they produce uncontrolled and possibly undesirable inputs within their full range of actuation. A system is resilient if it can reach its target despite such a malfunction \cite{TAC}.

Robust control, which aims to drive the system to its target despite small disturbances, generally cannot synthesize an appropriate controller because we consider undesirable inputs that can have the same magnitude as the controls \cite{weak_robust_control, TAC}. In order to identify and counteract the faulty actuators, the controller relies on sensors measuring in real-time the outputs of each actuators \cite{actuators_measures}.

The resilience of systems with undesirable inputs has been studied in \cite{Heymann_long, Heymann_short} under the name \textit{max-min controllability}. Because these papers limit their theory to unbounded sets of control inputs, their elegant resilience conditions have limited applications.
The work \cite{Schmitendorf_MaxMin} addresses this issue and shows that bounding the inputs levies a death sentence to the simple max-min controllability condition of \cite{Heymann_long}. Similarly, the approaches of \cite{Schmitendorf, Schmitendorf_MaxMin, Delfour} all lead to overly complex conditions for resilience of control systems. 
A simpler approach to resilience comes from differential games theory with H\'ajek's duality theorems \cite{Hajek}. Based on these results and with the controllability conditions of \cite{Brammer}, we establish straightforward resilience conditions for general linear systems.

However, resilience only guarantees that the target is reachable despite the malfunction. It does not say how much longer the malfunctioning system needs to reach the target compared to the nominal undamaged system. If this delay is excessive, resilience becomes useless in practice. 
To measure the delay, quantitative resilience was introduced for driftless linear systems in \cite{SIAM_CT, Quantitative_Resilience} as the maximal ratio of the minimal times required to reach a target for the nominal and malfunctioning systems. However, the geometric approach developed for those works \cite{Maximax_Minimax_JOTA} does not translate to general linear systems.
Instead, in this paper we will use Lyapunov theory \cite{Kalman}.

To calculate quantitative resilience, we first need to determine the fastest time for a linear system to reach its target. Relying on Pontryagin's Maximum Principle \cite{Liberzon} numerous authors have established algorithms to determine time-optimal control signals \cite{Neustadt_numerical, Eaton, Ho, Fujisawa, Grognard_real, Grognard_complex}. Indeed, it is in general impossible to find an analytical expression for time-optimal control signals \cite{Athans}.

We also need to determine the fastest time for the malfunctioning system to reach its target under the undesirable input making that time the longest. The study of optimal controls of perturbed linear systems was initiated by LaSalle \cite{LaSalle} and pursued by Babunashvili \cite{Babunashvili}. The most fruitful approach to our second question comes once again from differential games \cite{Sakawa}, as it considers specifically the worst undesirable input instead of a generic perturbation.

Using these algorithms, we can evaluate the ratio of minimal times for the nominal and malfunctioning systems to reach a target $x_{goal}$ from an initial state $x_0$. Quantitative resilience was introduced in \cite{SIAM_CT} as the maximum of this ratio over all pairs $(x_0, x_{goal})$. As testing every possible pair is not meaningful, and the specific geometric structure of driftless systems considered in \cite{SIAM_CT} vanishes for general linear systems, we instead establish analytical bounds on the minimal reach times to approximate quantitative resilience. 
We are mostly interested by determining a lower bound $b$ to quantitative resilience, as in the worst case the malfunctioning system will need less than $1/b$ times longer than the nominal system to reach $x_{goal}$.

The main contributions of this work are twofold. Firstly, we establish simple necessary and sufficient conditions to verify the resilience of linear systems. Secondly, relying on Lyapunov theory, we establish analytical bounds on the quantitative resilience of linear systems.

The remainder of this paper is organized as follows. Section~\ref{sec:problem} states the two problems of interest. In Section~\ref{sec:prelim}, building on previous work, we derive necessary and sufficient resilience conditions. In Section~\ref{sec:quantitative}, we establish analytical bounds on the quantitative resilience of linear systems. The application of our work on a temperature control system is presented in Section~\ref{sec:results}.

\textit{Notation:} For a set $X \subseteq \mathbb{R}^n$, $co(X)$ denotes its convex hull, $\partial X$ its boundary and $\interior(X)$ its interior. 
The Pontryagin difference \cite{Hajek} between two sets $X \subseteq \mathbb{R}^n$, $Y \subseteq \mathbb{R}^n$ is denoted by $Z := X - Y = \{ z \in X : z + y \in X\ \text{for all}\ y \in Y\}$.
For a set $Z \subseteq \mathbb{C}$, we say that $Re(Z) \leq 0$ (resp. $Re(Z) = 0$) if the real part of each $z \in Z$ verifies $Re(z) \leq 0$ (resp. $Re(z) = 0$).
The set of eigenvalues of a matrix $A$ is $\lambda(A)$. If $\lambda(A) \subseteq \mathbb{R}$, then $\lambda_{min}^A$ and $\lambda_{max}^A$ are the minimal and maximal eigenvalues.
The nullspace and image of $A$ are $\ker(A)$ and Im$(A)$.
Positive definite matrix $P$ is denoted by $P \succ 0$ and defines a $P$\textit{-norm} as $\|x\|_P^2 := x^\top P x$.
The closed ball of the Euclidean norm $\| \cdot \|_2$ of radius $\varepsilon > 0$ and center $c \in \mathbb{R}^n$ is denoted by $\mathbb{B}_\varepsilon(c)$, and the unit sphere is $\mathbb{S} := \partial \mathbb{B}_1(0)$. For $x \in \mathbb{R}^n$, $\|x\|_\infty = \max |x_i|$.
The function $u : [0, +\infty) \rightarrow \mathcal{U}$ is alternatively denoted by $u(\cdot) \in \mathcal{U}$. If $B \in \mathbb{R}^{n \times m}$ and $\mathcal{U}$ is a set, then $B\mathcal{U} := \big\{ Bu : u \in \mathcal{U} \big\}$.
We also use $\mathbf{1} = (1,\hdots,1)$ to denote the vector of $1$.

\section{PROBLEM STATEMENT}\label{sec:problem}

We consider the linear time-invariant system
\begin{equation}\label{eq:initial ODE}
    \dot x(t) = Ax(t) + \bar{B} \bar{u}(t), \quad x(0) = x_0 \in \mathbb{R}^n, \quad \bar{u}(t) \in \bar{\mathcal{U}},
\end{equation}
with $A \in \mathbb{R}^{n \times n}$ and $\bar{B} \in \mathbb{R}^{n \times (m+p)}$ constant matrices and $\bar{\mathcal{U}}$ an hyperrectangle of $\mathbb{R}^{m+p}$.

After a loss of control authority over $p$ of the $m+p$ actuators, the input signal $\bar{u}(\cdot)$ is split between the undesirable inputs $w(\cdot) \in \mathcal{W}$ and the controlled inputs $u(\cdot) \in \mathcal{U}$, with the compact sets $\mathcal{U} \subseteq \mathbb{R}^m$ and $\mathcal{W} \subseteq \mathbb{R}^p$ such that $\mathcal{U} \times \mathcal{W} = \bar{\mathcal{U}}$.
Matrix $\bar{B}$ is accordingly split into two constant matrices $B \in \mathbb{R}^{n \times m}$ and $C \in \mathbb{R}^{n \times p}$ so that the dynamics become
\begin{equation}\label{eq:splitted ODE}
    \dot x(t) = Ax(t) + Bu(t) + Cw(t), \quad x(0) = x_0 \in \mathbb{R}^n,
\end{equation}
with $u(t) \in \mathcal{U}$ and $w(t) \in \mathcal{W}$.
\begin{definition}
    A target $x_{goal} \in \mathbb{R}^n$ is \textit{resiliently reachable} from $x_0 \in \mathbb{R}^n$ by system \eqref{eq:splitted ODE} if for all piecewise continuous $w(\cdot) \in \mathcal{W}$, there is $T > 0$ and piecewise continuous $u(\cdot) \in \mathcal{U}$ such that the solution to \eqref{eq:splitted ODE} exists, is unique and $x(T) = x_{goal}$.
\end{definition}

\begin{definition}
    System \eqref{eq:splitted ODE} is \textit{resilient} to the loss of the actuators corresponding to $C$ if every $x_{goal} \in \mathbb{R}^n$ is resiliently reachable from any $x_0 \in \mathbb{R}^n$ by system \eqref{eq:splitted ODE}.
\end{definition}

We are now led to our first problem.

\begin{problem}\label{prob:resilience}
    Given $A$, $B$, $C$, $\mathcal{U}$, $\mathcal{W}$, determine whether system \eqref{eq:splitted ODE} is resilient.
\end{problem}

If system \eqref{eq:splitted ODE} is indeed resilient, we know that it can reach any target despite malfunctions, but it might need an excessively long time, making its resilience useless in practice. To avoid this scenario, previous work \cite{SIAM_CT} established a practical metric quantifying the resilience of linear driftless systems. We now generalize this metric to any control system.
For a target $x_{goal} \in \mathbb{R}^n$, the \emph{nominal reach time} introduced in \cite{SIAM_CT} becomes
\begin{equation}\label{eq: def T_N^*}
    T_N^*(x_0, x_{goal}) := \hspace{-1mm} \underset{\bar{u}(\cdot)\, \in\, \bar{\mathcal{U}}}{\inf} \hspace{-1mm} \left\{  \hspace{-2mm} \begin{array}{c} T > 0 : x(T) = x_{goal}\\
    \text{following}\ \eqref{eq:initial ODE} \end{array} \hspace{-2mm} \right\},
\end{equation}
the \emph{malfunctioning reach time} becomes
\begin{equation}\label{eq: def T_M^*}
    T_M^*(x_0, x_{goal}) := \hspace{-2mm} \underset{w(\cdot)\, \in\, \mathcal{W}}{\sup} \hspace{-1mm} \left\{ \underset{u(\cdot)\, \in\, \mathcal{U}}{\inf} \hspace{-1mm} \left\{ \hspace{-2mm} \begin{array}{c}
         T > 0 : x(T) = x_{goal}\  \\
         \text{following}\ \eqref{eq:splitted ODE}
    \end{array} \hspace{-2mm} \right\} \hspace{-1mm} \right\},
\end{equation}
and the \emph{quantitative resilience} 
\begin{equation}\label{eq:r_q}
    r_q(x_{goal}) := \underset{x_0\, \in\, \mathbb{R}^n}{\inf}\ \frac{T_N^*(x_0, x_{goal})}{T_M^*(x_0, x_{goal})}.
\end{equation}
If $x_0 = x_{goal}$, then $T_N^* = T_M^* = 0$ and we take the convention that their ratio is $1$.
By definition, if \eqref{eq:splitted ODE} is resilient, then $T_M^*$ is finite and the $\sup$, $\inf$ in \eqref{eq: def T_M^*} become $\max$, $\min$ achieved with optimal signals $w^*$ and $u^*$. In that case, $T_N^*$ is also finite and $\inf$ becomes $\min$ in \eqref{eq: def T_N^*}.
We focus on the classical case $x_{goal} = 0$, so we write $T_N^*$, $T_M^*$ and $r_q$ without their argument $x_{goal}$.

Our resilience framework assumes that the controller has only access to the past and current values of $w$, but not to their future. Then, the optimal control $u^*$ of \eqref{eq: def T_M^*} cannot anticipate a random undesirable input $w^*$, and $T_M^*$ is not likely to be time-optimal as required by Definition 2.4 of \cite{SIAM_CT}.

The only way to calculate $u^*$ without any future knowledge of $w^*$ is to solve the intractable Isaac's main equation \cite{Borgest}. This equation is the differential games counterpart of the Hamilton-Jacobi-Bellman (HJB) equation. According to the review \cite{Isaacs_review}, Isaac's main equation is even more difficult to solve than the HJB equation, which usually results in intractable partial differential equations \cite{Liberzon}. Hence, \cite{Borgest} determines only suboptimal solutions in this setting, the paper itself concludes that its practical contribution is minimal.

Instead, we adopt the setting of \cite{Sakawa}, where $u^*$ and $w^*$ are unique, \textit{bang-bang} \cite{Rechtschaffen_equivalences}, and make the transfer from $x_0$ to $x_{goal}$ time-optimal.
The controller knows that $w^*$ will be chosen to make $T_M^*$ the longest. Thus, $u^*$ is chosen to react optimally to this worst undesirable input. Then, $w^*$ is chosen, and to make $T_M^*$ the longest, it is the same as the controller had predicted. Hence, from an outside perspective it appears as if the controller built $u^*$ knowing $w^*$ in advance, as reflected by \eqref{eq: def T_M^*}.
Then, $T_M^*$ is time-optimal and can be compared with $T_N^*$, leading to the following problem.

\begin{problem}\label{prob:r_q}
    Evaluate the quantitative resilience of system \eqref{eq:splitted ODE}.
\end{problem}

\section{PRELIMINARIES}\label{sec:prelim}

Following \cite{Hajek}, we introduce the dual system to \eqref{eq:splitted ODE}
\begin{equation}\label{eq: hajek simplification}
    \dot y(t) = Ay(t) + z(t), \qquad y(0) = x_0, \qquad z(t) \in \mathcal{Z},
\end{equation}
$\mathcal{Z} := B\mathcal{U} - (-C\mathcal{W}) = \big\{ z \in B\mathcal{U} : z - Cw \in B\mathcal{U}\ \text{for all}\ w \in \mathcal{W} \big\}$. If $0 \in \mathcal{W}$, then $z \in \mathcal{Z}$ if and only if for all $w \in \mathcal{W}$, there exists $u \in \mathcal{U}$ such that $z - Cw = Bu$. Informally, $z$ represents the control available after counteracting any undesirable input.

System \eqref{eq: hajek simplification} is related to Problem~\ref{prob:resilience} through the following duality theorem of \cite{Hajek}, whose proof is extended to non-zero states $x_{goal}$ in \cite{Rechtschaffen_equivalences}, and reformulated in \cite{Schmitendorf}.

\vspace{1mm}

\begin{theorem}[H\'ajek Duality theorem]\label{thm:hajek reciprocity}
    The state of \eqref{eq:splitted ODE} can be driven to $x_{goal}$ at time $T$ by control signal $\phi$ for all $w(\cdot) \in \mathcal{W}$, if and only if the state of \eqref{eq: hajek simplification} can be driven to $x_{goal}$ at time $T$ by a control signal $z(\cdot) \in \mathcal{Z}$. Furthermore, $\phi(w,t) = Bu(t) = z(t) - Cw(t)$.
\end{theorem}

\vspace{1mm}

Theorem~\ref{thm:hajek reciprocity} transforms the problem of resilient reachability to one of \textit{bounded controllability}. Using Corollary 3.7 of \cite{Brammer} we thus trivially obtain the following resilience condition.

\vspace{1mm}
\begin{theorem}\label{thm: N&S resilience}
    If $0 \in \mathcal{Z}$ and $\interior(co(\mathcal{Z})) \neq \emptyset$, then system \eqref{eq:splitted ODE} is resilient if and only if $Re\big(\lambda(A)\big) = 0$ and there is no real eigenvector $v$ of $A^\top$ satisfying $v^\top z \leq 0$ for all $z \in \mathcal{Z}$.
\end{theorem}

\vspace{1mm}
The first two conditions of Theorem~\ref{thm: N&S resilience} ensure that the controls can move the system in any direction of the state space. For any target to be resiliently reachable, the controls need to be always able to overcome the drift term, so the eigenvalues of $A$ must have a zero real part.
Let us deduce a simpler condition for resilience.

\vspace{1mm}

\begin{cor}\label{cor:sufficient}
    If $0 \in \interior(\mathcal{Z})$ and $Re\big(\lambda(A)\big) = 0$, then system \eqref{eq:splitted ODE} is resilient.
\end{cor}
\begin{proof}
    Obviously $0 \in \mathcal{Z}$ and $\interior(co(\mathcal{Z})) \neq \emptyset$ because $\mathcal{Z} \subseteq co(\mathcal{Z})$. Let $v$ be any eigenvector of $A^\top$. Since $0 \in \interior(\mathcal{Z})$, there exists $\varepsilon > 0$ such that $\mathbb{B}_\varepsilon(0) \subseteq \mathcal{Z}$. Take $z = \varepsilon \frac{v}{\|v\|}$, then $v^\top z = \varepsilon \|v\| > 0$. Then, all the conditions of Theorem~\ref{thm: N&S resilience} are satisfied, so system \eqref{eq:splitted ODE} is resilient.
\end{proof}

\vspace{1mm}

Corollary~\ref{cor:sufficient} states that if the controls are stronger than all undesirable inputs, and if the drift does not increase exponentially, then the system is resilient.
For driftless systems, as studied in \cite{SIAM_CT} and \cite{Quantitative_Resilience}, Corollary~\ref{cor:sufficient} simplifies to

\vspace{1mm}

\begin{cor}\label{cor:driftless}
    Let $A = 0$ and assume that $0 \in \interior(\mathcal{Z})$. Then, driftless system \eqref{eq:splitted ODE} is resilient.
\end{cor}

\vspace{1mm}

Similarly, if a driftless system is resilient, then $-C\mathcal{W} \subseteq \interior(B\mathcal{U})$ and $\rank(B) = n$ \cite{Quantitative_Resilience}. We show that this condition is also sufficient for resilience due to the compactness of $\bar{\mathcal{U}}$.

\vspace{1mm}

\begin{proposition}\label{prop:compactness}
     \hspace{-2mm} If \hspace{-1mm} $-C\mathcal{W} \hspace{-0.2mm} \subseteq \hspace{-.5mm} \interior(B\mathcal{U})$ \hspace{-1.2mm} and \hspace{-1mm} $\rank(B) = n$, then $0 \hspace{-.3mm} \in \hspace{-.3mm} \interior(\mathcal{Z})$ \hspace{-2mm}.
\end{proposition}
\begin{proof}
    The compactness of $\bar{\mathcal{U}}$ implies that of $\mathcal{U}$ and $\mathcal{W}$, as well as of $B\mathcal{U}$ and $C\mathcal{W}$.
    By assumption, $-C\mathcal{W}$ and $\partial B\mathcal{U}$ are disjoint. Since they are both compact, there is a positive distance between them \cite{measure}, in the sense that there exists $\delta > 0$ such that $(-C\mathcal{W}) + \mathbb{B}_\delta(0) \subseteq B\mathcal{U}$, because $\rank(B) = n$. Thus, for all $w \in \mathcal{W}$ and all $z \in \mathbb{B}_\delta(0)$, $-Cw + z \in B\mathcal{U}$, i.e., $z \in \mathcal{Z}$. Then, $\mathbb{B}_\delta(0) \subseteq \mathcal{Z}$, i.e., $0 \in \interior(\mathcal{Z})$.
\end{proof}

\vspace{1mm}

Proposition~\ref{prop:compactness} simplifies Corollary~\ref{cor:sufficient} as the set $\mathcal{Z}$ is more difficult to characterize than $B\mathcal{U}$ and $C\mathcal{W}$. The price of this simplification is that the system needs full actuation.
We will now focus on the classical case of resilience to $x_{goal} = 0$, as in \cite{Sakawa}.

\vspace{1mm}

\begin{definition}
     System \eqref{eq:splitted ODE} is \emph{resiliently stabilizable} if for all $x_0 \in \mathbb{R}^n$ and all piecewise continuous $w(\cdot) \in \mathcal{W}$, there exists $T > 0$ and piecewise continuous $u(\cdot) \in \mathcal{U}$ such that the solution to \eqref{eq:splitted ODE} exists, is unique and $x(T) = 0$.
\end{definition}

\vspace{1mm}

Using Corollary 3.6 of \cite{Brammer} we trivially transform Theorem~\ref{thm:hajek reciprocity} into the following resilient stability condition.

\vspace{1mm}
\begin{theorem}\label{thm: N&S resiliently stabilizable}
    If $0 \in \mathcal{Z}$ and $\interior(co(\mathcal{Z})) \neq \emptyset$, then system \eqref{eq:splitted ODE} is resiliently stabilizable if and only if $Re \big(\lambda(A) \big) \leq 0$ and there is no real eigenvector $v$ of $A^\top$ satisfying $v^\top z \leq 0$ for all $z \in \mathcal{Z}$.
\end{theorem}

\vspace{1mm}

As before, a simpler sufficient condition can be derived with an analoguous proof.

\vspace{1mm}

\begin{cor}\label{cor:resilience to zero}
    If $0 \in \interior(\mathcal{Z})$ and $Re\big( \lambda (A) \big) \leq 0$, then system \eqref{eq:splitted ODE} is resiliently stabilizable.
\end{cor}

\vspace{1mm}


We have now established several conditions to verify the resilience of a linear system: Problem~\ref{prob:resilience} is solved. 
We now proceed to the more challenging question of the quantitative resilience of our system.

\section{QUANTITATIVE RESILIENCE}\label{sec:quantitative}

We consider unit bounded controls, $\bar{\mathcal{U}} = [-1,1]^{m+p}$, in line with previous works \cite{LaSalle, Eaton, Ho, Fujisawa, Grognard_real, Grognard_complex, Romano}. Note that scaling column $i$ of $\bar{B}$ enables a straightforward extension to symmetric input sets: $\bar{u}_i(t) \in [-\bar{u}_i^{max}, \bar{u}_i^{max}]$.

Unlike with driftless systems \cite{SIAM_CT}, the nominal and malfunctioning reach times $T_N^*$ and $T_M^*$ are not homogeneous in $x_0$, i.e., $T_N^*(\alpha x_0) \neq |\alpha| T_N^*(x_0)$ for $\alpha \in \mathbb{R}$. Indeed, in the driftless case $x(T) = x_0 + T \bar{B}\bar{u}$ for a constant $\bar{u}$, while for general linear systems $x(T) = e^{AT}x_0 + \int_0^T e^{A(T-t)} \bar{B}\bar{u}(t) dt$ is not linear in $T$. Thus, the optimization domain of the infimum in \eqref{eq:r_q} cannot be scaled down to $\mathbb{S}$ as in \cite{SIAM_CT}.

The loss of homogeneity in $x_0$ makes the determination of $r_q$ a much harder task than in the driftless case, which already required considerable amounts of calculations \cite{Quantitative_Resilience, Maximax_Minimax_JOTA}. We start by investigating the nominal reach time.

\subsection{Nominal reach time}

In order to calculate $r_q$, we need $T_N^*(x_0)$ for all $x_0 \in \mathbb{R}^n$. However, $T_N^*(x_0)$ has no general closed-form solution \cite{Athans}, and it cannot be computed for all $x_0 \in \mathbb{R}^n$. Thus, we establish analytical bounds on $T_N^*$ so that we can approximate $r_q$.

\vspace{1mm}

\begin{proposition}\label{prop: T_N lb b_max}
    If \eqref{eq:initial ODE} is stabilizable and $A$ is Hurwitz, then
    \begin{equation}\label{eq: T_N lb b_max}
        T_N^*(x_0)\ \geq\ 2\frac{\lambda_{min}^P}{\lambda_{max}^Q}\ln \Bigg(1 + \frac{\lambda_{max}^Q \|x_0\|_P}{2 \lambda_{min}^P b_{max}^P }\Bigg),
    \end{equation}
    for any $P \succ 0$ and $Q \succ 0$ such that $PA + A^\top P = -Q$ and with $b_{max}^P := \max\big\{ \|\bar{B}\bar{u}\|_P : \bar{u} \in \bar{\mathcal{U}}\big\}$.
\end{proposition}
\begin{proof}
    Since $A$ is Hurwitz, by the Lyapunov theorem \cite{Kalman}, for any $Q \succ 0$ there is $P \succ 0$ such that $PA + A^\top P = -Q$. Let us consider any such pair $(P,Q)$. We define the Lyapunov function $V(x) := x^\top P x = \|x\|_P^2$.
    Then,
    \begin{align*}
        \dot V(x) &= \dot x^\top P x + x^\top P \dot x = x^\top (A^\top P + PA)x + 2 x^\top P\bar{B}\bar{u} \\
        &= -x^\top Q x + 2 x^\top P\bar{B}\bar{u}.
    \end{align*}
    Since $P \succ 0$, there exists $M \in \mathbb{R}^{n \times n}$ such that $P = M^\top M$ \cite{Matrix}. Then, $x^\top P \bar{B}\bar{u} = (Mx)^\top M \bar{B} \bar{u} \geq - \|Mx\|_2 \|M\bar{B} \bar{u}\|_2$, by the Cauchy-Schwarz inequality \cite{Matrix}. 
    Note that $\|Mx\|_2^2 = x^\top M^\top M x = x^\top P x = \|x\|_P^2$. Similarly, $\|M \bar{B} \bar{u}\|_2 = \|\bar{B} \bar{u}\|_P$.
    
    The maximum in $b_{max}^P$ exists since $\bar{\mathcal{U}}$ is compact and the map $:\bar{u} \mapsto \|\bar{B}\bar{u}\|_P$ is continuous. Since $Q \succ 0$, we have $x^\top Qx \leq \lambda_{max}^Q \|x\|_2^2$ and $\|x\|_2^2 \leq \|x\|_P^2 / \lambda_{min}^P$ because $P \succ 0$. Finally, for $x \neq 0$ we obtain
    \begin{equation}\label{eq:V dot lb}
        \dot V(x) = \frac{d}{dt} \|x\|_P^2 \geq -\frac{\lambda_{max}^Q}{\lambda_{min}^P} \|x\|_P^2 - 2 b_{max}^P \|x\|_P.
    \end{equation}
    Let $y := \|x\|_P$, $\alpha := \frac{\lambda_{max}^Q}{2\lambda_{min}^P} > 0$ and $\beta := b_{max}^P > 0$. We divide \eqref{eq:V dot lb} by $2y$ so that $\dot y \geq w(y) := -\alpha y - \beta$.
   Then, inspired from the proof of the Bihari inequality in \cite{inequalities}, for $r \geq r_0 := \frac{-\beta}{1+\alpha}$, define
    \begin{align*}
        G(r) &:= \int_{r_0}^r \frac{ds}{w(s)} = \int_{r_0}^r \frac{-ds}{\alpha s + \beta} = \frac{-1}{\alpha} \Big[ \ln (\alpha s + \beta) \Big]_{r_0}^{r} \\
        &= \frac{-1}{\alpha} \ln\left( \frac{\alpha r + \beta}{\alpha r_0 + \beta}\right).
    \end{align*}
    The integral is well-defined since $\alpha s + \beta > 0$ for $s \in [r_0, r]$.
    Note that $\frac{d}{dt} G\big( y(t) \big) = \frac{\dot y(t)}{w( y(t) )} \leq 1$ because $\dot y(t) \geq w(y(t))$ and $w(y(t)) < 0$ since $\alpha > 0$ and $\beta > 0$. Thus, for any $T > 0$,
    \begin{equation}\label{eq:G lb}
        G\big( y(T) \big) - G\big( y(0) \big) = \int_0^T \frac{d}{dt} G\big( y(t) \big)\, dt \leq \int_0^T \hspace{-2mm} 1\, dt = T.
    \end{equation}
    
    Because $\bar{\mathcal{U}}$ is compact and convex, and \eqref{eq:initial ODE} is stabilizable, there exists a time-optimal control signal $\bar{u}^*(\cdot) \in \bar{\mathcal{U}}$ driving the state from $x_0$ to the origin in time $T_N^*(x_0)$ \cite{Liberzon}. Then, applying \eqref{eq:G lb} to this trajectory yields
    \begin{align*}
        T_N^*(x_0) &\geq \frac{-1}{\alpha} \ln\left( \frac{\alpha \|0\|_P + \beta}{\alpha r_0 + \beta}\right) + \frac{1}{\alpha} \ln\left( \frac{\alpha \|x_0\|_P + \beta}{\alpha r_0 + \beta}\right) \nonumber \\
        &\geq \frac{1}{\alpha} \ln \left( \hspace{-1mm} 1 + \frac{\alpha}{\beta}\|x_0\|_P \hspace{-1mm} \right) = 2\frac{\lambda_{min}^P}{\lambda_{max}^Q}\ln \Bigg(1 + \frac{\lambda_{max}^Q \|x_0\|_P}{2 \lambda_{min}^P b_{max}^P }\Bigg).
    \end{align*}
\end{proof}

We now want to find an upper bound to $T_N^*(x_0)$ following a similar procedure. 

\vspace{1mm}

\begin{proposition}\label{prop: T_N ub b_min}
    If $\rank(\bar{B}) = n$ and $A$ is Hurwitz, then 
    \begin{equation}\label{eq: T_N ub b_min}
        T_N^*(x_0)\ \leq\ 2\frac{\lambda_{max}^P}{\lambda_{min}^Q}\ln \Bigg(1 + \frac{\lambda_{min}^Q \|x_0\|_P}{2 \lambda_{max}^P b_{min}^P }\Bigg),
    \end{equation}
    for any $P \succ 0$ and $Q \succ 0$ such that $PA + A^\top P = -Q$ and with $b_{min}^P := \min\big\{ \| \bar{B}\bar{u} \|_P : \bar{u} \in \partial\bar{\mathcal{U}} \big\}$.
\end{proposition}
\begin{proof}
    Since $A$ is Hurwitz, by the Lyapunov theorem, for any $Q \succ 0$ there exists $P \succ 0$ such that $PA + A^\top P = -Q$. Let $V(x) := x^\top P x$. As in the proof of Proposition~\ref{prop: T_N lb b_max}, $\dot V(x) = -x^\top Q x + 2 x^\top P\bar{B}\bar{u}$. The minimum in $b_{min}^P$ exists since the map $: \bar{u} \mapsto \|\bar{B}\bar{u}\|_P$ is continuous and $\partial \bar{\mathcal{U}}$ is compact. Because $\rank(\bar{B}) = n$, we can choose $\bar{B}\bar{u}(t) = -\frac{x(t)}{\|x(t)\|_P} b_{min}^P$ for $x(t) \neq 0$.
    Indeed, assume for contradiction purposes that $\bar{u} \notin \bar{\mathcal{U}}$. Then, $\|\bar{u}\|_\infty > 1$. Let $\hat{u} := \frac{\bar{u}}{\|\bar{u}\|_\infty}$. Then, $\| \hat{u} \|_\infty = 1$, so $\hat{u} \in \partial \bar{\mathcal{U}}$, but $\| \bar{B} \hat{u}\|_P = \frac{\|\bar{B}\bar{u}\|_P}{\|\bar{u}\|_\infty} = \frac{b_{min}^P}{\|\bar{u}\|_\infty} < b_{min}^P$, which is a contradiction. Hence, the proposed control signal is admissible.
    However, it is not always time-optimal and thus only leads to an upper bound of $T_N^*(x_0)$. We obtain $2x^\top P\bar{B}\bar{u} = -2 b_{min}^P \|x\|_P$, so that
    \begin{equation}\label{eq:V dot ub}
        \frac{d}{dt} \|x\|_P^2 = \dot V(x) \leq -\frac{\lambda_{min}^Q}{\lambda_{max}^P} \|x\|_P^2 - 2 b_{min}^P \|x\|_P.
    \end{equation}
    Let $y := \|x\|_P$, $\alpha := \frac{\lambda_{min}^Q}{2\lambda_{max}^P} > 0$ and $\beta := b_{min}^P > 0$. For $x \neq 0$, dividing \eqref{eq:V dot ub} by $2y > 0$, yields $\dot y \leq w(y) := -\alpha y - \beta < 0$.
    
    Since $w(y) < 0$, we have $\frac{\dot y(t)}{w( y(t) )} \geq 1$. With the same definition of $G$ as in the proof of Proposition~\ref{prop: T_N lb b_max} and similarly to \eqref{eq:G lb}, we obtain $G\big( y(T) \big) - G\big( y(0) \big) \geq \int_0^T 1\, dt = T$ for any $T > 0$. 
    
    We need to prove that our control law ensures a finite time convergence to $0$. Note that $\dot V \leq -2\alpha V -2\beta \sqrt{V} := r(V)$, and that $v(t) = e^{-2\alpha t} \left( \frac{\beta}{\alpha}(1 - e^{\alpha t}) + \|x_0\|_P \right)^2$ is a solution to $\dot v = r(v)$ with $v(0) = \|x_0\|_P^2$. At $\tau = \frac{1}{\alpha} \ln\left(1 + \frac{\alpha}{\beta}\|x_0\|_P\right)$, we have $v(\tau) = 0$. Since $\alpha > 0$ and $\beta > 0$, $v$ converges in finite time to $0$. Thus, according to Proposition~1 of \cite{Finite_time}, $x$ converges in a finite time $T$ to the origin with the control law $\bar{B}\bar{u} = \frac{-x}{\|x\|_P} b_{min}^P$.
    
    Since this control law is not time-optimal, $T \geq T_N^*(x_0)$. With the expression of $G$ calculated previously, we obtain $\frac{1}{\alpha} \ln\left( \hspace{-1mm} 1 + \frac{\alpha}{\beta} \|x_0\|_P \hspace{-1mm} \right) \geq T$. Substituting $\alpha$ and $\beta$ yields \eqref{eq: T_N ub b_min}.
 \end{proof}

\vspace{1mm}

\begin{remark}
    It is more complicated to calculate $b_{min}^P$ than its counterpart $b_{max}^P$ since the maximum of $\| \bar{B} \bar{u} \|_P$ is reached on a vertex of $\bar{\mathcal{U}}$, while its minimum is not. An algorithm for determining $b_{min}^P$ is given in \cite{polyhedron_algo}.
\end{remark}

Now that we have bounded the nominal reach time $T_N^*$, we can investigate the malfunctioning reach time $T_M^*$.

\subsection{Malfunctioning reach time}

In this section we will bound $T_M^*$ following similar methods we applied to $T_N^*$. The difference is that $T_N^*$ is solution to a minimization \eqref{eq: def T_N^*}, while $T_M^*$ is solution of a maximin optimization \eqref{eq: def T_M^*}.

\vspace{1mm}

\begin{proposition}\label{prop: T_M lb}
    If \eqref{eq:splitted ODE} is resiliently stabilizable and $A$ is Hurwitz, then 
    \begin{equation}\label{eq: T_M lb z_max}
        T_M^*(x_0)\ \geq\ 2\frac{\lambda_{min}^P}{\lambda_{max}^Q}\ln \Bigg(1 + \frac{\lambda_{max}^Q \|x_0\|_P}{2 \lambda_{min}^P z_{max}^P }\Bigg),
    \end{equation}
    for any $P \succ 0$ and $Q \succ 0$ such that $PA + A^\top P = - Q$ and with $z_{max}^P := \max\big\{ \| z \|_P : z \in \mathcal{Z} \big\}$.
\end{proposition}
\begin{proof}
    Since $B\mathcal{U}$ and $C\mathcal{W}$ are compact, $\mathcal{Z}$ is compact \cite{Pontryagin_difference}, so $z_{max}^P$ exists.
    Since $A$ is Hurwitz, for $Q \succ 0$ there exists $P \succ 0$ such that $A^\top P + PA = -Q$. Let $V(x) := x^\top Px$. Then, $\dot V(x) = -x^\top Q x + 2 x^\top P(Bu+Cw)$. Since \eqref{eq:splitted ODE} is resiliently stabilizable, we can take $w^*(\cdot)$ and $u^*(\cdot)$ to be the optimizers in \eqref{eq: def T_M^*}. 
    For $w^*(\cdot) \in \mathcal{W}$, the control signal $Bu^*(\cdot) \in B\mathcal{U}$ drive the state of \eqref{eq:splitted ODE} to $0$. Then, according to Theorem~\ref{thm:hajek reciprocity}, $z(\cdot) = Cw^*(\cdot) + Bu^*(\cdot) \in \mathcal{Z}$. Then, $\|Cw^*(t) + Bu^*(t)\|_P \leq z_{max}^P$, which yields 
   \begin{equation*}
       \dot V(x) \geq -\frac{\lambda_{max}^Q}{\lambda_{min}^P} \|x\|_P^2 - 2 z_{max}^P \|x\|_P.
   \end{equation*}
    We now proceed as in the second half of the proof of Proposition~\ref{prop: T_N lb b_max} to obtain \eqref{eq: T_M lb z_max}.
\end{proof}

\vspace{1mm}

Similarly, we upper bound the malfunctioning reach time.

\vspace{1mm}

\begin{proposition}\label{prop: T_M ub}
    If $0 \in \interior(\mathcal{Z})$ and $A$ is Hurwitz, then
    \begin{equation}\label{eq: T_M ub z_min}
        T_M^*(x_0)\ \leq\ 2\frac{\lambda_{max}^P}{\lambda_{min}^Q}\ln \Bigg(1 + \frac{\lambda_{min}^Q \|x_0\|_P}{2 \lambda_{max}^P z_{min}^P }\Bigg),
    \end{equation}
    for any $P \succ 0$ and $Q \succ 0$ such that $PA + A^\top P = - Q$ and with $z_{min}^P := \min\big\{ \| z \|_P : z \in \partial \mathcal{Z} \big\}$.
\end{proposition}
\begin{proof}
    Since $\mathcal{Z}$ is compact, so is $\partial \mathcal{Z}$ and thus $z_{min}^P$ exists.
    Since $A$ is Hurwitz, for $Q \succ 0$ there exists $P \succ 0$ such that $A^\top P + PA = -Q$. Let $V(x) := x^\top Px$. Then, $\dot V(x) = -x^\top Q x + 2 x^\top P(Bu+Cw)$. Let $w^*(\cdot)$ be the argument of the supremum in \eqref{eq: def T_M^*}, which is a maximum since the system is resiliently stabilizable according to Corollary~\ref{cor:resilience to zero}. 
    
    Because $0 \in \interior(\mathcal{Z}) \neq \emptyset$, we have $\mathbb{B}_{z_{min}^P}(0) \subseteq \mathcal{Z}$. Then, we can choose the control signal $z(t) := \frac{-x(t)}{\|x(t)\|_P} z_{min}^P \in \mathcal{Z}$.
    By definition of $\mathcal{Z}$, for any $w(\cdot) \in \mathcal{W}$ and thus for $w^*(\cdot) \in \mathcal{W}$ there exists $u(\cdot) \in \mathcal{U}$ such that $z(t) = Cw^*(t) + Bu(t)$. Then, applying $w^*$ and $u$ will lead to an upper bound of $T_M^*$ as the control is not necessarily optimal, while the undesirable input is optimal. Then,
    \begin{equation*}
        \dot V(x) \leq -\frac{\lambda_{min}^Q}{\lambda_{max}^P} \|x\|_P^2 - 2 z_{min}^P \|x\|_P.
    \end{equation*}
    We now proceed as in the second half of the proof of Proposition~\ref{prop: T_N ub b_min} to obtain \eqref{eq: T_M ub z_min}.
\end{proof}

\vspace{1mm}

We are now able to bound $T_N^*(x_0)/T_M^*(x_0)$ for all $x_0$ and thus obtain an approximate of quantitative resilience.

\subsection{Bounding Quantitative Resilience}

As mentioned in Section~\ref{sec:intro}, we are primarily interested in determining a lower bound to quantitative resilience. Indeed, $r_q \geq b$ implies that in the worst case, the malfunctioning system will take less than $1/b$ times longer than the nominal system to reach the origin from the same initial state.

\vspace{1mm}

\begin{theorem}\label{thm: lb rq}
    If $0 \in \interior(\mathcal{Z})$ and $A$ is Hurwitz, then
    \begin{equation}\label{eq: lb r_q}
        r_q \geq \max\left( \frac{\lambda_{min}^P \lambda_{min}^Q}{\lambda_{max}^P \lambda_{max}^Q},\ \frac{z_{min}^P}{b_{max}^P} \right),
    \end{equation}
    for any $P \succ 0$ and $Q \succ 0$ such that $A^\top P + PA = -Q$.
\end{theorem}
\begin{proof}
    Since $\interior(\mathcal{Z})$ is not empty in $\mathbb{R}^n$, $\dim(\mathcal{Z}) = n$. Because $\mathcal{Z} \subseteq B\mathcal{U} \subseteq \mathbb{R}^n$, we also have $\rank(B) = n$, hence $\rank(\bar{B}) = n$. With $A$ being Hurwitz, system \eqref{eq:initial ODE} is then stabilizable. Thus, we can use Propositions~\ref{prop: T_N lb b_max} and \ref{prop: T_M ub}.
    For $x_0 \in \mathbb{R}^n$, $x_0 \neq 0$, \eqref{eq: T_N lb b_max} and \eqref{eq: T_M ub z_min} yield
    \begin{equation*}
        \frac{T_N^*(x_0)}{T_M^*(x_0)} \geq \frac{\lambda_{min}^P \lambda_{min}^Q}{\lambda_{max}^P \lambda_{max}^Q} \frac{\ln \Big(1 + \frac{\lambda_{max}^Q \|x_0\|_P}{2 \lambda_{min}^P b_{max}^P }\Big)}{\ln \Big(1 + \frac{\lambda_{min}^Q \|x_0\|_P}{2 \lambda_{max}^P z_{min}^P }\Big)} := f( \|x_0\|_P).
    \end{equation*}
    Then, according to \eqref{eq:r_q}, $r_q \geq \underset{x_0\, \in\, \mathbb{R}^n}{\inf} f(\|x_0\|_P)$.
    To alleviate the notation, we define the positive constants $a := \frac{\lambda_{min}^P \lambda_{min}^Q}{\lambda_{max}^P \lambda_{max}^Q}$, $b := \frac{\lambda_{max}^Q}{2\lambda_{min}^P b_{max}^P}$ and $c := \frac{\lambda_{min}^Q}{2\lambda_{max}^P z_{min}^P}$, so that $f(s) = a \frac{\ln(1 + b s)}{\ln(1 + c s)}$. 
    
    If $b = c$, then $f(s) = a$ for all $s \geq 0$, so $r_q \geq a$. If $b > c$, then $f$ is increasing in $s$, so $\inf \big\{ f(s) : s > 0\big\} = \underset{s \rightarrow 0}{\lim}\, f(s)$. Using the first order expansion of $\ln$ for small $s$, we have
    \begin{equation*}
        f(s) = a \frac{\ln(1 + b s)}{\ln(1 + c s)} \approx a \frac{b s + o(s)}{c s + o(s)} = a \frac{b}{c} + o(1).
    \end{equation*}
    Then, $f(0) = a \frac{b}{c} = \frac{z_{min}^P}{b_{max}^P} > a$. 
    If $c > b$, then $f$ is decreasing, so $\inf \big\{ f(s) : s \geq 0\big\} = \underset{s \rightarrow +\infty}{\lim} f(s)$. Note that
    \begin{equation*}
        \frac{\ln(1 + b s)}{\ln(1 + c s)} = \frac{1 + \frac{\ln(\frac{1}{s} + b)}{\ln(s)} }{ 1 + \frac{\ln(\frac{1}{s} + c)}{\ln(s)} } \xrightarrow[s \rightarrow +\infty]{} \frac{1 + \frac{\ln(b)}{+\infty}}{ 1 + \frac{\ln(c)}{+\infty}} = 1.
    \end{equation*}
    Then, $\underset{s \rightarrow +\infty}{\lim} f(s) = a > a \frac{b}{c}$, so $\underset{s\, \geq\, 0}{\inf} f(s) = \max \big( a,\, a\frac{b}{c} \big)$.
\end{proof}

\vspace{1mm}

We can upper bound $r_q$ using a similar approach

\vspace{1mm}

\begin{theorem}\label{thm: ub rq}
    If $\rank(\bar{B}) = n$, \eqref{eq:splitted ODE} is resiliently stabilizable and $A$ is Hurwitz, then
     \begin{equation}\label{eq: ub r_q}
        r_q \leq \max\left( \frac{\lambda_{max}^P \lambda_{max}^Q}{\lambda_{min}^P \lambda_{min}^Q},\ \frac{z_{max}^P}{b_{min}^P} \right),
    \end{equation}
    for any $P \succ 0$ and $Q \succ 0$ such that $A^\top P + PA = -Q$.
\end{theorem}
\begin{proof}
    Following our assumptions, we are allowed to use Propositions~\ref{prop: T_N ub b_min} and \ref{prop: T_M lb}.
    For $x_0 \in \mathbb{R}^n$, $x_0 \neq 0$, combining \eqref{eq: T_N ub b_min} and \eqref{eq: T_M lb z_max} yields
    \begin{equation*}
        \frac{T_N^*(x_0)}{T_M^*(x_0)} \leq \frac{\lambda_{max}^P \lambda_{max}^Q}{\lambda_{min}^P \lambda_{min}^Q} \frac{\ln \Big(1 + \frac{\lambda_{min}^Q \|x_0\|_P}{2 \lambda_{max}^P b_{min}^P }\Big)}{\ln \Big(1 + \frac{\lambda_{max}^Q \|x_0\|_P}{2 \lambda_{min}^P z_{max}^P }\Big)} := g( \|x_0\|_P).
    \end{equation*}
    Then, according to \eqref{eq:r_q}, $r_q \leq \underset{x_0\, \in\, \mathbb{R}^n}{\inf} g(\|x_0\|_P)$.
    To alleviate the notation, we define the positive constants $a := \frac{\lambda_{max}^P \lambda_{max}^Q}{\lambda_{min}^P \lambda_{min}^Q}$, $b := \frac{\lambda_{min}^Q}{2\lambda_{max}^P b_{min}^P}$ and $c := \frac{\lambda_{max}^Q}{2\lambda_{min}^P z_{max}^P}$, so that $g(s) = a \frac{\ln(1 + b s)}{\ln(1 + c s)}$.
    
    This function $g$ is similar to $f$ in the proof of Theorem~\ref{thm: lb rq}, and thus $r_q \leq \underset{x_0\, \in\, \mathbb{R}^n}{\inf} g(\|x_0\|_P) = \max \big( a,\, a\frac{b}{c} \big)$ yields \eqref{eq: ub r_q}.
\end{proof}

Note that we used the same pair $(P,Q)$ to bound both $T_N^*$ and $T_M^*$. Employing different pairs $(P_N, Q_N)$ and $(P_M, Q_M)$ would make $f$ depend on both $\|x_0\|_{P_N}$ and $\|x_0\|_{P_M}$. Then, we would need to take $x_0 \in \mathbb{R}^n$ instead of $\|x_0\|_P \in \mathbb{R}^+$ as the argument of $f$, which would significantly complicate the minimum search.
We leave this more convoluted approach for possible future work.

\section{NUMERICAL RESULTS}\label{sec:results}

We illustrate our work on a room temperature control system motivated by \cite{temperature} and schematized in Figure~\ref{fig:rooms}.

\begin{figure}[htbp!]
    \centering
    \begin{tikzpicture}[scale = 0.7]
        \draw[] (-4.5, 2) -- (-3, 2);
        \draw[] (3, 2) -- (4.5,2);
        \draw[] (-4.5, 0) -- (-3, 0);
        \draw[] (3, 0) -- (4.5,0);
        \node at (-4, 1) {$T_{goal}$};
        \node at (4, 1) {$T_{goal}$};
        \draw[] (-4.5, -1) -- (4.5, -1);
        \node at (3.5, -0.7) {\textcolor{red}{hallway}};
        
        \node at (4, 3) {\textcolor{blue}{outside}};
        \node at (-4.2, 3) {\textcolor{orange}{Sun}};
        
        \draw[thick] (-2.3, 0) -- (-3, 0) -- (-3, 2) -- (-2.3, 2) -- (-2.1, 2.2);
        \draw[thick] (-1.9, 2.2) -- (-1.7, 2) -- (-1, 2) -- (-1, 0) -- (-1.7, 0) -- (-2.1, -0.4);
        \node at (-2, 1) {$T_1$};
        \draw[blue, ultra thick, <-] (-2, 2.6) -- (-2, 1.8);
        \node at (-2, 2.8) {\textcolor{blue}{$q_{w1}$}};
        \draw[red, ultra thick, ->] (-2.5, -0.4) -- (-1.9, 0.2);
        \node at (-2.6, -0.6) {\textcolor{red}{$q_{d1}$}};
        \draw[ultra thick, <->] (-3.3, 1) -- (-2.7, 1);
        \node at (-2.6, 0.6) {\textcolor{black}{$q_{g1}$}};
        \draw[orange, ->] (-3.2, 2.5) -- (-2.7, 2);
        \draw[orange, ->] (-3.3, 2.5) -- (-2.8, 2);
        \draw[orange, ->] (-3.1, 2.5) -- (-2.6, 2);
        \node at (-3.2, 2.8) {\textcolor{orange}{$q_{S1}$}};
        \draw[blue, ->] (-1.4, 1.8) -- (-1.4, 2.2);
        \draw[blue, ->] (-1.6, 1.8) -- (-1.6, 2.2);
        \draw[blue, ->] (-1.2, 1.8) -- (-1.2, 2.2);
        \node at (-1.4, 1.6) {\textcolor{blue}{$q_{l1}$}};

        \draw[thick] (-0.3, 0) -- (-1, 0) -- (-1, 2) -- (-0.3, 2) -- (-0.1, 2.2);
        \draw[thick] (0.1, 2.2) -- (0.3, 2) -- (1, 2) -- (1, 0) -- (0.3, 0) -- (-0.1, -0.4);
        \node at (0, 1) {$T_2$};
        \draw[blue, ultra thick, <-] (0, 2.6) -- (0, 1.8);
        \node at (0, 2.8) {\textcolor{blue}{$q_{w2}$}};
        \draw[red, ultra thick, ->] (-0.5, -0.4) -- (0.1, 0.2);
        \node at (-0.6, -0.6) {\textcolor{red}{$q_{d2}$}};
        \draw[ultra thick, <->] (-1.3, 1) -- (-0.7, 1);
        \node at (-0.6, 0.6) {\textcolor{black}{$q_{12}$}};
        \draw[orange, ->] (-1, 2.5) -- (-0.5, 2);
        \draw[orange, ->] (-1.1, 2.5) -- (-0.6, 2);
        \draw[orange, ->] (-0.9, 2.5) -- (-0.4, 2);
        \node at (-1, 2.8) {\textcolor{orange}{$q_{S2}$}};
        \draw[blue, ->] (0.4, 1.8) -- (0.4, 2.2);
        \draw[blue, ->] (0.6, 1.8) -- (0.6, 2.2);
        \draw[blue, ->] (0.8, 1.8) -- (0.8, 2.2);
        \node at (0.6, 1.6) {\textcolor{blue}{$q_{l2}$}};

        \draw[thick] (1.7, 0) -- (1, 0) -- (1, 2) -- (1.7, 2) -- (1.9, 2.2);
        \draw[thick] (2.1, 2.2) -- (2.3, 2) -- (3, 2) -- (3, 0) -- (2.3, 0) -- (1.9, -0.4);
        \node at (2, 1) {$T_3$};
        \draw[blue, ultra thick, <-] (2, 2.6) -- (2, 1.8);
        \node at (2, 2.8) {\textcolor{blue}{$q_{w3}$}};
        \draw[red, ultra thick, ->] (1.5, -0.4) -- (2.1, 0.2);
        \node at (1.4, -0.6) {\textcolor{red}{$q_{d3}$}};
        \draw[ultra thick, <->] (3.3, 1) -- (2.7, 1);
        \node at (1.4, 0.6) {\textcolor{black}{$q_{23}$}};
        \draw[ultra thick, <->] (1.3, 1) -- (0.7, 1);
        \node at (3.4, 0.6) {\textcolor{black}{$q_{3g}$}};
        \draw[orange, ->] (1, 2.5) -- (1.5, 2);
        \draw[orange, ->] (1.1, 2.5) -- (1.6, 2);
        \draw[orange, ->] (0.9, 2.5) -- (1.4, 2);
        \node at (1, 2.8) {\textcolor{orange}{$q_{S3}$}};
        \draw[blue, ->] (2.4, 1.8) -- (2.4, 2.2);
        \draw[blue, ->] (2.6, 1.8) -- (2.6, 2.2);
        \draw[blue, ->] (2.8, 1.8) -- (2.8, 2.2);
        \node at (2.6, 1.6) {\textcolor{blue}{$q_{l3}$}};
        
    \end{tikzpicture}
    \caption{Scheme of the rooms and of the heat transfers. The heater $q_h$ and AC transfers $q_{AC}$ are not shown for clarity.}
    \label{fig:rooms}
\end{figure}
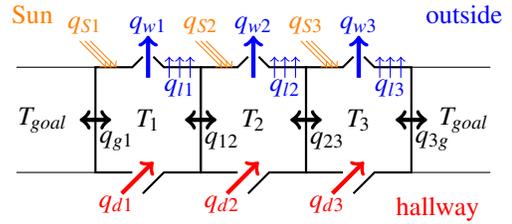

The objective is to make the rooms 1, 2 and 3 reach the target temperature $T_{goal}$, which is the temperature of the two neighboring rooms as represented on Figure~\ref{fig:rooms}. There are several ways to control the temperature. The central conditioning unit controls the heater $q_h$ and the AC $q_{AC}$ in all the rooms. Outside of work hours, the controller can also incrementally open doors $q_d$ and windows $q_w$ for room specific adjustments. Because each room is fitted with large windows and blinds, one can also use solar heating $q_S$. The heat loss through the outside wall is $q_l$ and the heat transfer between rooms $i$ and $j$ is $q_{ij}$.
Following \cite{temperature}, the temperature dynamics of the rooms are
\begin{align*}
    mC_p \dot T_1 &= (q_h - q_{AC}) + (q_{d1} - q_{w1}) + (q_{S1} - q_{l1}) + q_{g1} + q_{12}, \\
    mC_p \dot T_2 &= (q_h - q_{AC}) + (q_{d2} - q_{w2}) + (q_{S2} - q_{l2}) - q_{12} + q_{23}, \\
    mC_p \dot T_3 &= (q_h - q_{AC}) + (q_{d3} - q_{w3}) + (q_{S3} - q_{l3}) - q_{23} + q_{g3},
\end{align*}
with $m$ the mass of air in each room, $C_p$ is the specific heat capacity of air, $q_{gi} = aU_{gi}(T_{goal} - T_i)$, $q_{ij} = aU_{ij}(T_j - T_i)$, with $a$ the area of the wall between rooms and $U_{ij}$ the overall heat transfer coefficient between rooms $i$ and $j$. This coefficient depends on the wall materials, which are made of concrete and wood. 

Since most of the heat transfers cannot be modeled with symmetric inputs, we combine them in pairs: $q_h - q_{AC} =: Q_{hAC} u_{hAC}$, $q_{di} - q_{wi} =: Q_{dw} u_{dw}^i$ and $q_{Si} - q_{li} =: Q_{Sl} u_{Sl}^i$ with $u_{hAC} \in [-1, 1]$, $u_{dw}^i \in [-1, 1]$, $u_{Sl}^i \in [-1, 1]$ and $i \in \{1,2,3\}$.
The numerical values used in our calculations are gathered in Table~\ref{tab:values}.

\begin{table}[htbp!]
\caption{Numerical values for the simulation, based on \cite{temperature}.}
\begin{tabular}{ccc|ccc}
    Parameter & Value & Unit & Parameter & Value & Unit \\ \hline
    $a$ & $12$ & $m^2$ & $mCp$ & $42186$ & $J \cdot K^{-1}$ \\
    $U_{g1}$ & $6.27$ & $W \cdot K^{-1}$ & $U_{12}$ & $5.08$ & $W \cdot K^{-1}$ \\
    $U_{23}$ & $5.41$ & $W \cdot K^{-1}$ & $U_{3g}$ & $6.27$ & $W \cdot K^{-1}$ \\
    $Q_{hAC}$ & $350$ & $W$ & $Q_{dw}$ & $300$ & $W$ \\
    $Q_{Sl}$ & $200$ & $W$ & $T_{goal}$ & $293$ & $K$ \\
\end{tabular}
\label{tab:values}
\end{table}
We write the dynamics as $\dot T = AT + \bar{B}\bar{u} + D T_{goal}$, with
\begin{equation*}
    A = \frac{1}{mCp} \begin{pmatrix} -U_{g1} -U_{12} & U_{12} & 0 \\ U_{12} & -U_{12} - U_{23} & U_{23} \\ 0 & U_{23} & - U_{23} - U_{3g} \end{pmatrix},
\end{equation*}
\begin{equation*}
    \bar{B} = \frac{1}{mCp}\begin{pmatrix} Q_{Sl} & 0 & 0 & Q_{dw} & 0 & 0 & Q_{hAC} \\
                                           0 & Q_{Sl} & 0 & 0 & Q_{dw} & 0 & Q_{hAC} \\
                                           0 & 0 & Q_{Sl} & 0 & 0 & Q_{dw} & Q_{hAC}
    \end{pmatrix},
\end{equation*}
\begin{equation*}
    \bar{u}^\top = \begin{pmatrix} u_{Sl}^1 & u_{Sl}^2 & u_{Sl}^3 & u_{dw}^1 & u_{dw}^2 & u_{dw}^3 & u_{hAC} \end{pmatrix},
\end{equation*}
\begin{equation*}
    T = \begin{pmatrix} T_1 \\ T_2 \\ T_3 \end{pmatrix} \quad \text{and} \quad D = \frac{1}{mCp}\begin{pmatrix} U_{g1} \\ 0 \\ U_{3g} \end{pmatrix}.
\end{equation*}
Let $x^\top := \big( x_1\ x_2\ x_3 \big) = T - \mathbf{1} T_{goal}$. Then, $\dot x = \dot T = Ax + \bar{B}\bar{u} + D T_{goal} + A\mathbf{1}T_{goal} = Ax + \bar{B}\bar{u}$ and $x_{goal} = \big( 0\ 0\ 0 \big)$. 

Since $\lambda(A) = \big\{ -0.052, -0.033, -0.010\big\} \subseteq \mathbb{R}^-$, $A$ is Hurwitz. Then, according to Theorem~\ref{thm: N&S resilience}, the system cannot be resilient to the loss of any actuator, but it can be resiliently stabilizable by Theorem~\ref{thm: N&S resiliently stabilizable}. Note that for the loss of any one column $C$, we have $\rank(B) = 3$ and we can numerically verify that $-C\mathcal{W} \subseteq \interior(B\mathcal{U})$. Then, following Proposition~\ref{prop:compactness} and Corollary~\ref{cor:resilience to zero}, we know that the system is resiliently stabilizable.

We consider a situation where a worker remains in their room after hours and manually opens or closes their door and window, thus overriding the controller.
Let us quantify the resilience of the system to the loss of control over the door and window in room 1, i.e., over $u_{dw}^1$.

To use the upper and lower bounds derived in Section~\ref{sec:quantitative}, we need to choose positive definite matrices $P$ and $Q$ solving the Lyapunov equation $A^\top P + PA + Q = 0$. The simplest approach is to generate a stochastic $Q \succ 0$ and to solve the equation for $P \succ 0$.
Another approach relies on the fact that for $y$ small, $\ln(1+y) \approx y$. Thus, the lower bound of $T_M^*$ in \eqref{eq: T_M lb z_max} can be approximated by $\frac{\|x_0\|_P}{z_{max}^P}$. To maximize this lower bound, we minimize $z_{max}^P = \max\big\{ \|z\|_P : z \in \mathcal{Z}\big\}$, i.e., we choose $P \succ 0$ generating the tightest ellipsoid outer approximation of $\mathcal{Z}$. Similarly, to minimize the upper bound \eqref{eq: T_M ub z_min}, we need $P$ to generate the largest ellipsoid inside $\mathcal{Z}$. Then, we take $Q = -A^\top P - PA$, but there is no guarantee that $Q \succ 0$.

Since we shifted the coordinate system by $T_{goal} = 20^\circ C$, an initial state slightly warmer than desired can be $x_0^\top = \big( 0.8^\circ C\ 0.7^\circ C\ 0.9^\circ C \big)$.
We compute $T_N^*(x_0)$ with Eaton's algorithm \cite{Eaton} and its lower and upper bounds with \eqref{eq: T_N lb b_max} and \eqref{eq: T_N ub b_min}: $35.5s \leq T_N^*(x_0) = 42.5s \leq 54.1s$.
We compute $T_M^*(x_0)$ with Sakawa's algorithm \cite{Sakawa} and its lower and upper bounds with \eqref{eq: T_M lb z_max} and \eqref{eq: T_M ub z_min}: $53s \leq T_M^*(x_0) = 110.5s \leq 135s$.

Then, the room temperatures can take up to $T_M^*(x_0)/T_N^*(x_0) = 2.6$ times longer to all reach $T_{goal}$ after the loss of control authority over the door and window in room 1. Our bounds lead to a worst-case time increase by a factor of up to $3.8$. As can be seen on Figure~\ref{fig:T_M}, the bounds generated with the tight ellipsoidal approximations of $\mathcal{Z}$ are better than the bounds obtained with stochastic pairs $(P,Q)$.

\begin{figure}[htbp!]
    \centering
    \includegraphics[scale=0.5]{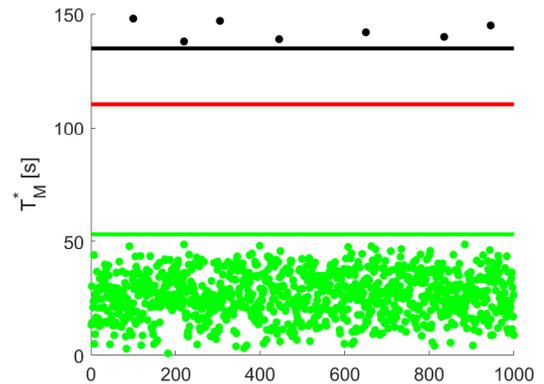}
    \caption{Bounds on the malfunctioning reach time $T_M^*(x_0)$ represented by the red line. The green and black dots are the lower bounds \eqref{eq: T_M lb z_max} and the upper bounds \eqref{eq: T_M ub z_min} for 1000 randomly generated solution pairs $(P,Q)$.
    The green and black line are the lower and upper bounds \eqref{eq: T_M lb z_max} and \eqref{eq: T_M ub z_min} generated with the ellipsoid approximations of $\mathcal{Z}$.}
    \label{fig:T_M}
\end{figure}

Using Theorem~\ref{thm: lb rq} and \ref{thm: ub rq} we can also bound the resilience of the system starting from any initial condition $x_0$ and we obtain $0.097 \leq r_q  \leq 2.79$. The lower bound means that the loss of control over $u_{dw}^1$ can make the damaged system up to $1/0.097 = 10.3$ times slower than the nominal system to reach $T_{goal}$ in the worst case. Since the upper bound is larger than $1$, it does not convey any information.

If the system loses control authority over the central heating/AC unit, i.e., $u_{hAC}$, the rooms can take as much as $T_M^*(x_0)/T_N^*(x_0) = 4.7$ times longer to reach $T_{goal}$ from the same initial temperature, while our bound predicts a ratio of up to $9.3$. These ratios are larger than for the loss of control over the window and door of room 1 because $Q_{hAC} > Q_{dw}$ and the central heating/AC affects directly all 3 rooms.

\section{CONCLUSION}

This paper explores and quantifies the resilience of control systems to the loss of control authority over some of their actuators.
We established novel necessary and sufficient conditions for the resilience of general linear systems. 
Based on Lyapunov theory and differential inequalities we derived analytical bounds on the nominal and malfunctioning reach times, allowing to approximate quantitative resilience of general linear systems.

There are several avenues of future work. We want to investigate thoroughly how to determine the pair of matrices $(P,Q)$ generating the best bounds on the reach times $T_N^*$ and $T_M^*$ and on the quantitative resilience $r_q$.
We also desire to extend our resilience theory to nonlinear systems and to more complex and more natural scenarios where the system must visit a succession of targets.
Ensuring the safety of critical systems by preventing them from visiting dangerous locations while completing their mission even after enduring a loss of control is also among our future objectives.

\addtolength{\textheight}{-0cm}   




\bibliographystyle{IEEEtran}
\bibliography{ref.bib}

\end{document}